%% file: main.tex
\documentclass[superscriptaddress,nofootinbib,tightenlines,preprint]{revtex4-1}

\usepackage{pacchetti}

\input{newcommands}

\DeclareUnicodeCharacter{2212}{-}

\date{\today}

\begin{document}

\title{A generalization of the Hawking black hole area theorem}

\author{Eleni-Alexandra Kontou}
\email{eleni.kontou@kcl.ac.uk}
\affiliation{ITFA and GRAPPA, Universiteit van Amsterdam, Science Park 904, Amsterdam, the Netherlands}
\affiliation{Department of Mathematics, King’s College London, Strand, London WC2R 2LS, United Kingdom}

\author{Veronica Sacchi}
\email{veronica.sacchi@epfl.ch}
\affiliation{Scuola Normale Superiore, Piazza dei Cavalieri 7, 56126 Pisa, Italy}
\affiliation{University of Pisa, Lungarno Antonio Pacinotti 43, 56126 Pisa, Italy}
\affiliation{École Polytechnique Fédéral de Lausanne (EPFL), Route de la Sorge, CH-1015 Lausanne, Switzerland}

\begin{abstract}
Hawking's black hole area theorem was proven using the null energy condition (NEC), a pointwise condition violated by quantum fields. The violation of the NEC is usually cited as the reason that black hole evaporation is allowed in the context of semiclassical gravity. Here we provide two generalizations of the classical black hole area theorem: First, a proof of the original theorem with an averaged condition, the weakest possible energy condition to prove the theorem using focusing of null geodesics. Second, a proof of an area-type result that allows for the shrinking of the black hole horizon but provides a bound on it. This bound can be translated to a bound on the black hole evaporation rate using a condition inspired from quantum energy inequalities. Finally, we show how our bound can be applied to two cases that violate classical energy conditions.

\end{abstract}

\maketitle

\tableofcontents

\newpage

\section{Introduction}
\label{sec:introduction}

Hawking proved his famous black hole area theorem in 1971 \cite{Hawking:1971vc}, showing that, under certain spacetime assumptions, the area of the black hole horizon can never decrease. This important result in classical general relativity paved the way to the laws of black hole thermodynamics: the area of the black hole horizon is a measure of its entropy and according to the second law it cannot decrease. It was again Hawking that a few years later proved that black holes emit radiation \cite{Hawking:1975vcx}, a result in tension with his previous theorem: as black holes emit radiation their mass and thus the area of their horizon decreases.   

To interpret this violation of the area theorem we need to look closely at its assumptions and in particular at the null energy condition (NEC). Energy conditions are bounds on components of the stress-energy tensor and were introduced in general relativity as reasonable assumptions for any matter model \cite{kontou2020energy,curiel2017primer}. They have been used as assumptions in a variety of classical relativity theorems including the singularity theorems \cite{penrose1965gravitational,hawking1966occurrence}. The NEC states that
\be
T_{\mu \nu}U^\mu U^\nu \geq 0 \,,
\ee
were $U^\mu$ is a null vector field. This condition is obeyed by most reasonable classical theories. However, this and all pointwise conditions are violated by quantum fields as Epstein, Glaser and Jaffe \cite{epstein1965nonpositivity} showed. 

Hawking radiation requires a semiclassical setting, quantum fields on a classical curved background. Thus the usual interpretation for the failure of the area theorem in a semiclassical setup is the violation of the NEC. But the negative energy that expectation values of the stress-energy tensor can admit is not unbounded in the context of a quantum field theory (QFT), at least when their averages are considered. Ford \cite{Ford:1978qya} first introduced the concept of a quantum energy inequality (QEI), interestingly as a way to prevent the violation of the second law of black hole thermodynamics. In general, QEIs bound the magnitude and duration of any negative energy densities or fluxes within a QFT. The most famous example is the bound on the renormalized energy density averaged over a segment of a timelike geodesic in Minkowski spacetime \cite{Ford:1994bj,Fewster:1998pu}
\be
\int_\gamma f^2(t) \langle \nord{\rho} \rangle dt \geq -\frac{1}{16\pi^2} \|f''(t)\|^2 \,,
\ee
where $\nord{\rho}$ is the normal ordered energy density. For $f$ a smooth compactly supported real-valued function, the lower bound is finite. QEIs have been derived for free fields and a few interacting ones for flat and curved spacetimes (see \cite{Fewster2017QEIs,kontou2020energy} for reviews). 

In recent years, QEIs have been used to prove semiclassical singularity theorems. The idea is to first replace the pointwise energy conditions of the classical theorems with a condition of the form of QEIs \cite{Fewster:2010gm,fewster2020new}. Then show that there exist quantum fields that obey such a condition \cite{Fewster:2021mmz,Freivogel:2020hiz}. We should note that the singularity theorems (as the area theorem) use a geometric condition, the null (or timelike) convergence condition
\be
R_{\mu \nu}U^\mu U^\nu \geq 0 \,,
\ee
for $U^\mu$ null (or timelike). Classically, one can use the Einstein equation to go from the NEC to the null convergence condition and back. Semiclassically, what is used to connect the expectation values of the renormalized stress-energy tensor with classical curvature is the semiclassical Einstein equation (SEE)
\be
\label{eqn:see}
G_{\mu \nu}=8\pi \langle T^{\text{ren}}_{\mu \nu} \rangle_\omega \,.
\ee
A complete solution to that equation should include a set of a metric $g_{\mu \nu}$ and a quantum state $\omega$. But solutions are notoriously difficult to find and their existence has only been shown in highly symmetric cases (e.g. \cite{Gottschalk:2018kqt,Meda:2020smb,Sanders:2020osl}). Its use in the generalization of classical relativity theorems is much simpler: one uses the SEE directly to connect the QEIs with a geometric condition.

The original singularity theorems as well as the area theorem used the Raychaudhuri equation to show convergence of the congruence of geodesics in their proofs. Some generalizations of singularity theorems with weakened energy conditions also used the Raychaudhuri equation and properties of Riccati inequalities \cite{tipler1978energy,Borde:1987qr,Roman:1988vv,Fewster:2010gm}. More recently Fewster and Kontou \cite{fewster2020new} suggested the use of index form methods to prove singularity theorems with conditions inspired by QEIs.  This method is not new, it was used by O’Neill to prove the orginal singularity theorems \cite{o1983semi} and also by Chicone and Ehrlich \cite{chicone1980line} to prove the existence of conjugate points along complete geodesics using averaged energy conditions.

Fewster and Kontou used the method Ref.~\cite{fewster2020new} to prove the first semiclassical singularity theorem \cite{Fewster:2021mmz} for timelike geodesic incompleteness using the quantum strong energy inequality \cite{Fewster:2018pey} of a minimally coupled quantum scalar field. While Ref.~\cite{fewster2020new} includes the proofs of singularity theorems for both timelike and null geodesic incompleteneness with conditions inspired by QEIs, the null case presents significant difficulties. With a direct counterexample, Fewster and Roman \cite{Fewster:2002ne} showed that the renormalized expectation value of the null energy, averaged over a finite segment of a null geodesic is unbounded from below. To overcome that problem, Freivogel and Krommydas \cite{Freivogel:2018gxj} introduced a UV cutoff to control the lower bound in these cases. The QEI, called the smeared null energy condition (SNEC) has the form
\be
		\label{eqn:SNEC}
		 \int_\gamma f^2(\lambda)\langle \nord{T_{\mu\nu}U^{\mu}U^{\nu}} \rangle d\lambda \ge -\frac{4B}{G_N}\vert\vert f' \vert\vert^2 \,,
\ee
where $G_N \lesssim \ell_{\text{UV}}^{n-2}/N$ is the effective Newton's constant, $\ell_{\text{UV}}$, the UV cutoff of the theory and $N$ the number of quantum fields. The SNEC has been proven for free fields in Minkowski spacetime \cite{Fliss:2021gdz} and it was used to prove a semiclassical singularity theorem for null geodesic incompleteness \cite{Freivogel:2020hiz}.

In this work we apply the method of Ref.~\cite{fewster2020new} to the case of the area theorem. Unlike the singularity theorems, which seem to hold semiclassically, we do not provide a proof of the area theorem using QEI-inspired conditions. A proof like that would be in contrast to the concept of black hole evaporation. Instead the goal of this paper is two-fold: 
\begin{itemize}
\item
We find the weakest condition under which the classical area theorem proof holds 
\item
We derive a bound on the evaporation rate of semiclassical black holes using QEI-inspired conditions
\end{itemize}
The second result is effectively a semiclassical generalization of the area theorem. The interpretation is straightforward: QFT allows for negative energy and so the decrease of the black hole horizon area. However, the bounds on the amount of negative energy provide a bound to the rate of its decrease.

The paper is organized in the following manner. In Sec.~\ref{sec:index} we review the index form method for null geodesics. In Sec.~\ref{sec:classical} we prove the original Hawking black hole area theorem using the index form method and find the weakest energy condition under which it holds. In Sec.~\ref{sec:generalized} we prove the main result of our paper, a bound on the rate of evaporation of semiclassical black holes. In Sec.~\ref{sec:applications} we apply the theorem in two examples that violate the NEC, the classical non-minimally coupled scalar field and the quantum minimally coupled scalar using SNEC. In both cases we compare with the calculated evaporation rate. We conclude in Sec.~\ref{sec:discussion} and provide two additional results on the location of the trapped surface in semiclassical black holes in Appendix~\ref{app:trapped} and the allowed duration of negative energy in Appendix~\ref{app:l0upperbound}.

\vspace{0.5in} 

\textit{Conventions:} We use the $(-,+,+)$ sign convention from Misner, Thorne and Wheeler \cite{misner1973gravitation}. The units are $c=G=1$ except in part of Sec.~\ref{sec:applications} and we work in $n$ spacetime dimensions unless otherwise stated. 

\section{The index form method}
\label{sec:index}

In this section we briefly review the notion of the index form for null geodesics and how it is used to prove the formation of focal points. We loosely follow Ref.~\cite{o1983semi}.

\subsection{Variation of the action integral}
\label{sub:action}

To begin, as we want to study null geodesics that have zero proper time, we study the \textit{action} or \textit{energy} functional. Given a curve \(\gamma : [0, \ell] \rightarrow M\) affinely parameterized by \(\lambda\), \(E\) is defined as
\be
E[\gamma] \coloneqq \frac{1}{2}\int_{0}^{l} g(\gamma'(\lambda), \gamma'(\lambda))d\lambda \,.
\ee
Let \(P\) be a spacelike submanifold of \(M\) of co-dimension \(2\) and consider the set of all piecewise smooth curves joining \(P\) to \(q\), \(\Omega(P, q)\). Then the family of curves \(\gamma_s(\lambda) \coloneqq
\zeta(\lambda, s)\) in \(\Omega(P, q)\), varies smoothly in $s$. The tangent and the transverse vector fields are defined as $U_\mu=\gamma'(\lambda)$ and $V_\mu=\partial \gamma_s/\partial s |_{s=0}$. The second variation of $E[\gamma_s]$ or \textit{Hessian} $\mathcal{H}[V]$ is found to be \cite{o1983semi, fewster2020new}
\be
	\label{eq:hessian}
	\mathcal{H}[V]\equiv \frac{\partial^2E[\gamma_s]}{\partial s^2}\Big\vert_{s = 0} = 
	\int_{0}^{\ell} \left[(\nabla_UV_{\mu})(\nabla_UV^{\mu}) + R_{\mu\nu\alpha\beta}U^{\mu}V^{\nu}V^{\alpha}U^{\beta}\right] d\lambda -U^{\mu} \nabla_V\nabla_UV_{\mu} \Big\vert_0^l\,.
\ee
The second variation of the length functional, used for timelike geodesics is called the \textit{index form}, from which the method takes its name.

Let now \(e_i\) with \(i = 1, \ldots, n - 2\) be an orthonormal basis of \(T_{\gamma(0)}P\), and parallel transport it along \(\gamma\) to generate \(\{E_i\}_{i = 1, \ldots, n-2}\). Then, take \(f\) a smooth function with \(f(0) = 1\) and \(f(l) = 0\). Calculating the Hessian for $fE_i$ gives
\begin{equation}
	\mathcal{H}(fE_i, fE_i) = \int_{0}^{\ell} \left(-f'(\lambda)^2 + f(\lambda)^2R_{\mu\nu\alpha\beta}U^{\mu}E_i^{\nu}E_i^{\alpha}U^{\beta}\right) d\lambda- U_{\mu}\mathrm{I\!I}^{\mu}(E_i, E_i)\Big\vert_{\gamma(0)}\,,
\end{equation}
where $\mathrm{I\!I}$ is the shape tensor or second fundamental form. We then sum over all \(i = 1, \ldots, n - 2\) to get:
\begin{equation}
	\label{eq:hessian-averagded}
	\sum_{i=1}^{n - 2}\mathcal{H}(fE_i, fE_i) = - \int_{0}^{\ell} \left((n - 2)f'(\lambda)^2 - f(\lambda)^2R_{\mu\nu}U^{\mu}U^{\nu}\right) d\lambda - (n - 2)U_{\mu}\mathrm{H}^{\mu}\Big\vert_{\gamma(0)}.
\end{equation}
Here 
\be
H^\mu=\frac{1}{n-2}\sum_{i=1}^{n-2} \mathrm{I\!I}^{\mu}(E_i, E_i) \,,
\ee
is the mean normal curvature vector field of $P$.

\subsection{Formation of focal points}

A focal point on a causal geodesic is defined formally as follows
\begin{definition}
	Let \(\gamma\) be a causal geodesic of \(M\) normal to \(P \subset M\). Then \(\gamma(r)\), where \(r \neq 0\) is a \emph{focal point} of \(P\) along \(\gamma\) provided there is a nonzero \(P\)-Jacobi field \(V\) on \(\gamma\), with \(V(r) = 0\).
\end{definition}
More informally, a focal point of a submanifold \(P\) along a normal geodesic \(\gamma\) is an \emph{almost}-meeting point of nearby \(P\)-normal geodesics of the same causal character of \(\gamma\).

For timelike geodesics, after a focal point, the curves cease to extremize length. For null geodesics, to have an analogous statement we need the notion of \textit{promptness}. The term was introduced by Witten\cite{witten2020light} who defined a causal path from $q$ to $p$ as ``prompt'' if there is no causal path from $q$ to a point $r$ near $p$ and just to its past. 

\noindent More informally, there is no causal path, starting from \(q\), that arrives sooner to $p$. The notion is similar when we have null geodesics emanating normally from a spacelike hypersurface. As in the case of timelike geodesics, a null geodesic is not prompt after a focal point. 

To determine the existence or not of a focal point we use the Hessian calculated in subsection~\ref{sub:action}. In particular
\begin{prop}
	\label{prop:H-positivity-criteria}
	If there are no focal points of \(P \subset M\) along a normal null geodesic \(\gamma\in\Omega(P,q)\), then \(\mathcal{H}[V]\) is negative semidefinite on \(T_{\gamma}^{\perp}\Omega = \{V \in T_{\gamma}\Omega \wedge V \perp \gamma\}\). Furthermore if \(V \in T_{\gamma}^{\perp}\Omega \) and \(\mathcal{H}[V] = 0\) then \(V\) is tangent to \(\gamma\)\,.
\end{prop}
The proof of Proposition~\ref{prop:H-positivity-criteria} can be found in Ref.~\cite{o1983semi} (Proposition 10.41). Then using Eq.~\eqref{eq:hessian-averagded} the proof of the following proposition is immediate.
\begin{prop}
		\label{prop:fp-criteria}
		Let \(P\) be a spacelike submanifold of \(M\) of co-dimension \(2\) and let \(\gamma\) be a null geodesic joining \(p \in P\) to \(q\in J^+(P)\). Let $\gamma$ be affinely parametrized by $\lambda \in [0,\ell]$.  If there exist a smooth \((-\frac{1}{2})\)-density \(f\) which is non vanishing at \(p\) but is null at \(q\), and such that
		\begin{equation}
		\label{eq:fp-criteria}
		\int_0^\ell \big((n -2)f'(\lambda)^2 - f(\lambda)^2 R_{\mu \nu} U^\mu U^\nu \big)d\lambda \le -(n -2) U_\mu H^\mu \big|_{\gamma(0)}
		\end{equation}
		then there is a focal point to \(P\) along \(\gamma\). If the inequality holds strictly then the focal point is located before \(q\).
	\end{prop}
For a discussion on the invariance of Eq.~\eqref{eq:fp-criteria} see Ref.~\cite{fewster2020new}. Before we proceed we should address the issue of parametrization of null geodesics. Following Ref.~\cite{fewster2020new} we fix the affine parameter $\lambda$ by requiring 
\be
\hat{H}_\mu \frac{d\gamma^\mu}{d\lambda}=1 \,,
\ee
where $\hat{H}^\mu$ is a unit timelike vector defined as $H^\mu=H \hat{H^\mu}$ where $H^\mu$ is the mean normal curvature of $P$. Now we can state the main result of this section, the classical focusing theorem (see Ref.~\cite{o1983semi} Prop.10.43).

	\begin{theorem}[Classical Focusing theorem]
		\label{the:focusingcla}
		Let \(P\) be a spacelike submanifold of codimension \(2\). Choose $\gamma$ a null geodesic parametrized by $\lambda \in [0,\ell]$ such that $\hat{H}_\mu U^\mu=1$, where $U^\mu$ is the tangent vector to $\gamma$. Suppose that
		\begin{itemize}
		\item [(i)]
		$P$ is future converging, meaning that $H<0$;
		\item[(ii)]
	    the null-convergence condition
	    \be
	    R_{\mu \nu} U^\mu U^\nu \geq 0 \,,
	    \ee
	    holds everywhere on \(\gamma\).
	    \end{itemize}
	    Then if $\ell \geq 1/|H|$ there is a focal point on $\gamma$.
	\end{theorem}

\begin{proof}
In the chosen coordinates, pick $f(\lambda)=1-\lambda/\ell$. Then using the fact that the null convergence condition holds on $\gamma$ the inequality of Eq.~\eqref{eq:fp-criteria} is true if $\ell \geq 1/|H|$. Then $\gamma$ has a focal point before $\lambda= \ell$.
\end{proof}

\section{The classical black hole area theorem}
\label{sec:classical}

In this section we first present a proof of the original Hawking black hole area theorem using the index form method presented in Sec.~\ref{sec:index}. Then we prove the area theorem using the weakest possible condition for this method that does not allow for any horizon decrease.

\subsection{The Hawking black hole area theorem}
\label{subsec:classical-bh-area}

Here we follow the notation and part of the proof structure of Wald~\cite{wald2010general}. However, we note that both the original paper~\cite{Hawking:1971vc} and  Ref.~\cite{wald2010general} use the Raychaudhuri equation which is avoided here. 

\begin{theorem}[Black Hole Area Theorem]
	\label{th:classical-bh-area}
	Suppose that
	\begin{enumerate}
	    \item[(i)]
	 \((M, g_{\mu\nu})\) is a strongly asymptotically predictable spacetime;
	 \item[(ii)]
	 the null convergence condition
	 \be
	 R_{\mu\nu}U^{\mu}U^{\nu} \ge 0
	 \ee
	 holds for all null vectors \(U^{\mu}\).
	 \end{enumerate}
	 Let \(\Sigma_1\) and \(\Sigma_2\) be spacelike Cauchy surfaces for the globally hyperbolic region \(\tilde{V}\) such that \(\Sigma_2 \subset I^+(\Sigma_1)\), and given the event horizon \(H\) we define
	\be
	\mathscr{H}_1 = H \cap \Sigma_1\,, \quad \text{and}  \quad \mathscr{H}_2 = H \cap \Sigma_2 \,.
	\ee
	Then the area of \(\mathscr{H}_2\) is greater or equal than the area of \(\mathscr{H}_1\).
\end{theorem}

\begin{proof}
	Let \(\Sigma_1\) be any Cauchy hypersurface for \(\tilde{V}\) through \(p\) and  \(\mathrm{H}^{\mu}\) be the mean normal curvature vector field of \(\mathscr{H}_1\). We will prove by contradiction that for \(U^{\mu}\) the tangent field  of the null generators of the horizon \(H\), it holds everywhere that:
	\begin{equation}
	\label{eq:exp-null-generators}
		\mathrm{H}^{\mu}U_{\mu} \ge 0.
	\end{equation}
	Suppose instead that \(\mathrm{H}^{\mu}U_{\mu} < 0\) at \(p\in \mathscr{H}_1\). We then want to extend the function \(\mathrm{H}^{\mu}U_{\mu}\) in a continuous way on \(\Sigma_1\) in a neighborhood of $p$. We take any small deformation of \(\mathscr{H}_1\) outward on \(\Sigma_1\), say \(\mathscr{H}_1'\) and call \(K\) the closed region on \(\Sigma_1\) between \(\mathscr{H}_1\) and \(\mathscr{H}_1'\); the boundary of its future \(\partial J^+(K)\) is a null hypersurface of co-dimension \(1\), and hence comes with its own null generators, with tangent field \(U'^{\mu}\). This allows us to define the extended function in \(p' \in\mathscr{H}_1'\) as simply the contraction 
	\(\mathrm{H}'^{\mu}U'_{\mu}\) (with \(\mathrm{H}'^{\mu}\) the mean normal curvature of \(\mathscr{H}_1'\)). There are multiple possible extensions, but we only need the existsance of a smooth one. The deformed area is shown in Fig.~\ref{fig:extension}.
	
	\begin{figure}
		\centering
		\includegraphics[scale=0.5]{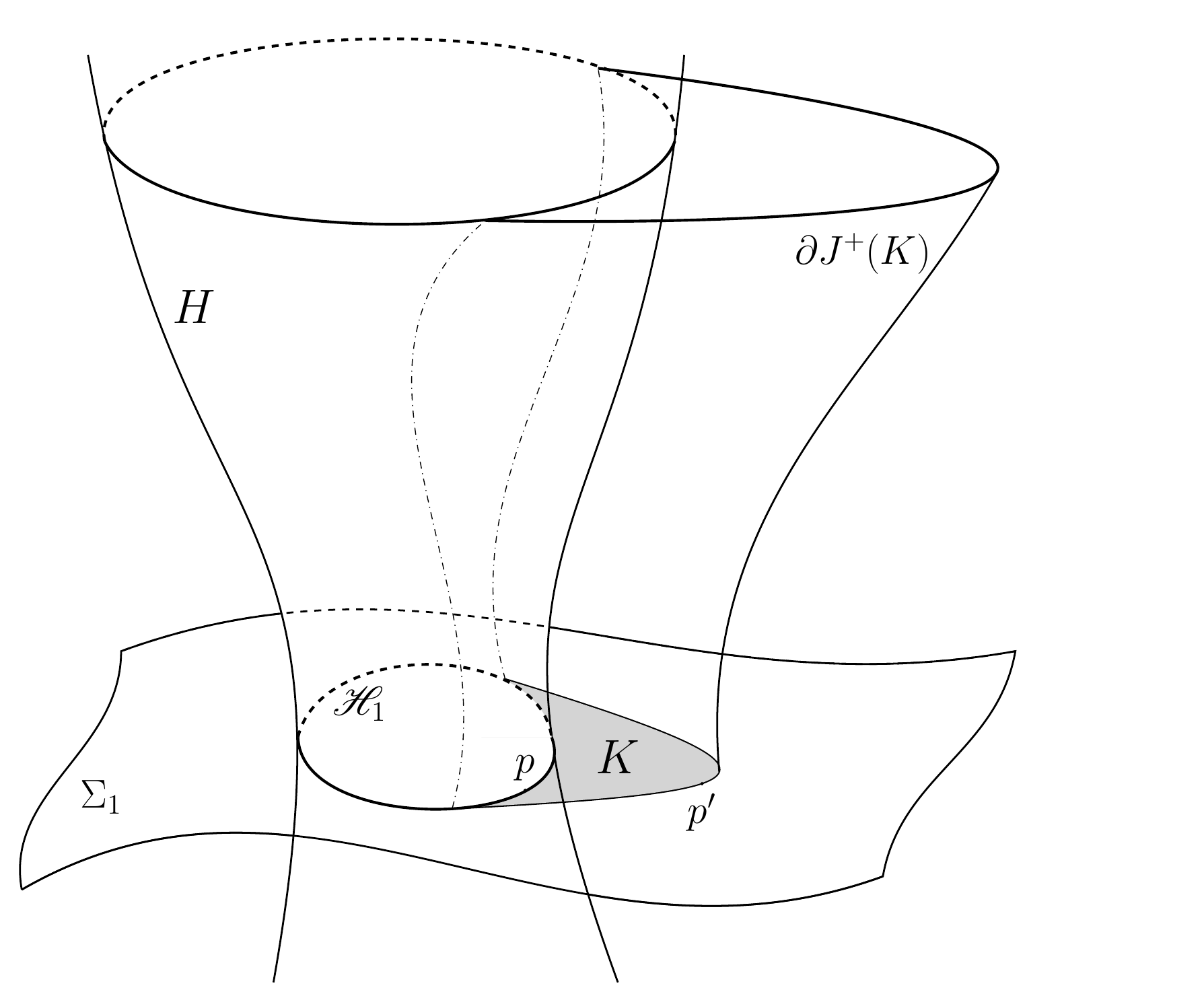}
		\caption{The smooth deformation of $\mathscr{H}_1$ on $\Sigma_1$. The shaded area $K$ is the closed region between \(\mathscr{H}_1\) and \(\mathscr{H}_1'\).}
			\label{fig:extension}
	\end{figure}
	Given that extension, there exists a neighborhood of \(p\) in which \(\mathrm{H'}^{\mu}U'_{\mu} < 0\). Then we can always choose a deformation \(\mathscr{H}_1\) outward on \(\Sigma_1\), to \(\mathscr{H}_1'\), such that
	\be
	\begin{cases}
	J^-(\mathscr{I}^+) \cap \mathscr{H}_1' \neq \emptyset; \\
	\mathrm{H'}^{\mu}U'_{\mu} < 0 \text{ everywhere on } 	J^-(\mathscr{I}^+) \cap \mathscr{H}_1' \,.
	\end{cases}
	\ee
	
	Next we choose a point \(q \in \mathscr{I}^+ \cap \partial J^+(K)\).
	The null geodesic generator through \(q\) will meet \(\mathscr{H}_1'\) orthogonally. Now we can apply Theorem~\ref{the:focusingcla} for $\gamma$ the null geodesic generator through $q$ with $\gamma(\lambda = 0) = p'$ and setting $\hat{\mathrm{H'}}^{\mu}U'_{\mu} = 1$. Since for $\mathscr{H'}_1$ we have $U'_{\mu}\mathrm{H}'^{\mu}=H'<0$, assumption (i) holds. Then a focal point will develop before $\ell=1/|H'|$ and so before $q$. This is impossible, because orthogonal null generators are prompt curves, and so cannot contain any focal points, hence 
	\be
	\forall p \in H\,, \quad \mathrm{H}^{\mu}U_{\mu}(p) \ge 0\,.
	\ee
		
	To conclude, it is enough to observe that each \(p\in \mathscr{H}_1\) lies on a null generator \(\gamma\) contained in \(H\). As \(\Sigma_2\) is a Cauchy hypersurface as well, \(\gamma\) must intersect \(\Sigma_2\) in a point \(q \in \mathscr{H}_2\). Then, the flow along null generators maps \(\mathscr{H}_1\) into a portion of \(\mathscr{H}_2\). But we know that under deformation along the flow of a vector field, the area of a submanifold evolves as (\cite{kriele1999spacetime} and \cite{senovilla2022critical})
	\begin{equation}
		\delta_UA_{\mathscr{H}_1} = \int_{\mathscr{H}_1} \mathrm{H}^{\mu}(p)V_{\mu} \ge 0 \,.
	\end{equation}
	Thus when we modify \(\mathscr{H}_1\) along the flow of the null generators the area of \(\mathscr{H}_1\) can never decrease.
\end{proof}

\subsection{A more general condition}

At the heart of the area theorem (as for singularity theorems) is the focusing theorem. The original Hawking area theorem - as we showed in the previous subsection - used the null convergence condition to deduce the formation of focal points of the null generators of the horizon. But using Prop.~\ref{prop:fp-criteria} we can easily prove the classical area theorem with a more general, averaged energy condition. In fact, this is the weakest condition for this method that does not allow the horizon area to decrease and so forbids black hole evaporation \footnote{Of course there may be weaker energy conditions by changing for example the causality or spacetime assumptions.}.

\begin{theorem}
\label{the:genclas}
	Suppose that
	\begin{enumerate}
	    \item[(i)]
	 \((M, g_{\mu\nu})\) is a strongly asymptotically predictable spacetime; 
	 \item[(ii)]
	 for \(H\) the event horizon, \(U^{\mu}\) the tangent field of its null generators $\gamma(\lambda)$, and \(\mathrm{H}^{\mu}\) the mean normal curvature of \(\mathscr{H} = \Sigma \cap H\), where \(\Sigma\) is a Cauchy surface for the globally hyperbolic region \(\tilde{V}\), the following inequality holds
	 \be
	 \label{eqn:classcond}
	 \int_\gamma f(\lambda)^2 R_{\mu \nu}U^\mu U^\nu d\lambda \geq (n-2) \|f'\|^2 \,,
	 \ee
	 where $f\in C^\infty_{1,0}[0,\ell]$ such that $f(0)=1$ and $f(\ell)=0$.
	 \end{enumerate}
	 Let \(\Sigma_1\) and \(\Sigma_2\) be spacelike Cauchy surfaces for the globally hyperbolic region \(\tilde{V}\) such that \(\Sigma_2 \subset I^+(\Sigma_1)\), and given the event horizon \(H\) we define
	\be
	\mathscr{H}_1 = H \cap \Sigma_1 \quad \quad \mathscr{H}_2 = H \cap \Sigma_2 \,.
	\ee
	Then the area of \(\mathscr{H}_2\) is greater or equal than the area of \(\mathscr{H}_1\).
\end{theorem}

\begin{proof}
The beginning of the proof is similar as in Theorem~\ref{th:classical-bh-area} but instead of Theorem~\ref{the:focusingcla} we apply Prop.~\ref{prop:fp-criteria} to deduce that there is a focal point. In particular using condition (ii) and Prop.~\ref{prop:fp-criteria} we conclude there is a focal point before $q$ if $H'^\mu U'_\mu <0$. Thus $\forall p \in H$, $H^\mu U_\mu(p) \geq 0$ and the proof is concluded.
\end{proof}

\begin{remark}
Assumption (ii) can be strengthened if instead of Eq.~\eqref{eqn:classcond} we use 
\be
\inf_{f\in C^\infty_{1,0}[0,\ell]} J_\ell[f] \leq 0 \,,
\ee
with 
\be
J_\ell[f]=\int_\gamma \left((n-2)f'(\lambda)^2 -f(\lambda)^2 R_{\mu \nu}U^\mu U^\nu \right) d\lambda \,.
\ee
\end{remark}

A special case of condition~\eqref{eqn:classcond} is the damped averaged null energy condition (dANEC) introduced by Lesourd~\cite{lesourd2018remark} who proved the area theorem using methods from Ref.~\cite{fewster2011singularity}. Both references use the Raychaudhuri equation. Picking $f=e^{-c\lambda/2}$ where $c>0$ Eq.~\eqref{eqn:classcond} for $n=4$ becomes
\be
\int_\gamma e^{-c\lambda} R_{\mu \nu}U^\mu U^\nu d\lambda-\frac{c}{2} \geq 0 \,,
\ee
which becomes the dANEC for future complete null geodesics \footnote{The half-ANEC and ANEC are usually conditions on the stress-energy tensor $T_{\mu \nu}$. The conditions on the curvature that are used in the singularity theorems and the area theorem can be obtained classically with the use of the Einstein Equation.}. It is interesting to note that for $c=0$ this condition reduces to half-ANEC 
\be
\int_\gamma R_{\mu \nu}U^\mu U^\nu d\lambda \geq 0 \,.
\ee
Even though those conditions are significantly weaker than the null convergence condition, we don't expect them to be satisfied by quantum fields. In general, the energy density (or the null energy) of quantum fields is not necessarily positive over portions of causal geodesics. Even the half-ANEC can be easily violated by having negative energy concentrated in one half of the geodesic, even though the ANEC is generally obeyed. Nevertheless, it is evident that it is not sufficient that quantum fields violate the NEC to allow black hole evaporation but instead they need to violate the average energy condition of Eq.~\eqref{eqn:classcond}.

\section{The generalized black hole area theorem}
\label{sec:generalized}

In the previous section we proved theorems where the black hole horizon area was not allowed to decrease. In this section we use the structure of the Hawking area theorem to provide a bound on the evaporation rate. We further show the structure of the bound for conditions of the form of QEIs.

\subsection{A bound on the evaporation rate}

Instead of requiring an energy condition to hold, we can use the causality assumption, which forbids the formation of focal points on null generators, to impose a bound on the black hole horizon area decrease. The following theorem directly uses Prop.~\ref{prop:fp-criteria}

\begin{theorem}
\label{th:general}
Suppose that \((M, g_{\mu\nu})\) is a strongly asymptotically predictable spacetime. Let \(H\) be the black hole horizon, \(U^{\mu}\) the tangent field of its null generators $\gamma(\lambda)$, and \(\mathrm{H}^{\mu}\) the mean normal curvature of \(\mathscr{H} = \Sigma \cap H\), where \(\Sigma\) is a Cauchy surface for the globally hyperbolic region $\tilde{V}$. Then the following inequality holds
\be
\label{eqn:horizonch}
		\delta_U\mathcal{A}_{\mathscr{H}} = \int_{\mathscr{H}} \mathrm{H}^{\mu}(p)U_{\mu} \ge - \frac{1}{n - 2}\left(\inf_{\substack{f\in C^{\infty}_{1,0}[0, \ell]}}J_{\ell}[f]\right)\cdot\mathcal{A}_{\mathscr{H}}.
\ee
where
\be
J_\ell[f]=\int_0^\ell \left((n-2)f'(\lambda)^2-f(\lambda)^2 R_{\mu \nu}U^\mu U^\nu  \right) d\lambda \,.
\ee
\end{theorem}

\begin{proof}
 By contradiction, suppose there exists a point \(p\) on the horizon \(H\) where
	\be
		U_{\mu}\mathrm{H}^{\mu} < -\frac{1}{n - 2} \inf_{\substack{f\in C^{\infty}_{1,0}[0, \ell]}}J_{\ell}[f] \,.
	\ee
	Similarly to the proof of Theorem~\ref{th:classical-bh-area} there exists a deformation of \(\mathscr{H}\) on the Cauchy surface \(\Sigma\) in a neighborhood of \(p\) where this inequality holds everywhere. For the null generators of \(\partial J^+(K)\), equation~\eqref{eqn:horizonch} implies that there exist \(\ell > 0\) and \(f\in C^{\infty}_{1,0}[0, \ell]\) so that condition~\(\eqref{eq:fp-criteria}\) is satisfied. But this leads to a contradiction in the same way as in theorem~\ref{th:classical-bh-area} because we are in a globally hyperbolic spacetime, and null generators are not allowed to contain any focal point. The proof is concluded by noting that the change of the black hole horizon area is 
	\be
		\delta_U\mathcal{A}_{\mathscr{H}} = \int_{\mathscr{H}} \mathrm{H}^{\mu}(p)U_{\mu} \,.
	\ee
\end{proof}
\begin{remark}
The rate of change of the horizon can be interpreted directly as a rate of change of the mass only in the case of Schwarzschild black holes. However, even in more general black holes, we expect a term in the expression for the horizon proportional to the mass of the black hole. In that sense, the rate of change of the horizon can be interpreted as a rate of change of the mass of the black hole. As this is a lower bound, it can then be interpreted as a bound on the evaporation rate.
\end{remark}

This straightforward result will be used to bound the rate of change of the area of the horizon for an energy condition that could be obeyed by quantum fields, as we examine in the next subsection.

\subsection{A condition inspired by QEIs}

Let $P$ be a spacelike submanifold of $M$ of co-dimension 2 with mean normal curvature vector field $H_\mu$. Suppose that $\gamma(\lambda)$ is a future-directed null geodesic emanating normally from $P$. Fix the parametrization of the affine parameter $\lambda$ so that $\hat{H}_{\mu} d\gamma^{\mu}/d\lambda=1$. Then for every smooth compactly supported $(−1/2)$-density $g$ on $\gamma$, and every choice of smooth non-negative constants $Q_0$ and $Q_m$ we assume
\be
\label{eqn:QEI}
\int_0^{\ell} g(\lambda)^2 R_{\mu\nu}U^{\mu}U^{\nu} \ge -Q_m(\gamma) \vert\vert g^{(m)}\vert\vert^2 - Q_0(\gamma) \vert\vert g\vert\vert^2 \,.
\ee
We would like to use this condition along with the condition for the formation of a focal point Eq.~\eqref{eq:fp-criteria} and specify the lower bound of the horizon area change of Eq.~\eqref{eqn:horizonch}. We should note that the two smooth functions $g$ and $f$ need different boundary conditions: for the bound of the energy condition of Eq.~\eqref{eqn:QEI} to be finite, $g$ needs to obey the generalised Dirichlet boundary conditions $g^{(k)}(0)=g^{(k)}(\ell)=0$ for all $0\leq k \leq m-1$. However, the condition for the formation of focal points Eq.~\eqref{eq:fp-criteria} requires $f(0)=1$ and $f(\ell)=0$.

To address that problem, we follow the method of Ref.~\cite{fewster2020new}. First, we pick a point $0<\ell_0<\ell$ on the geodesic. Then, we define two $C^\infty$ functions $f$ and $\varphi$ in the following way: $f(\lambda)=1$ for $0\leq\lambda<\ell_0$ and $f(\ell)=0$ while $\varphi(\lambda)=1$ for $\ell_0 \geq\lambda\geq \ell$ and $\varphi(0)=0$. So we can test Eq.~\eqref{eq:fp-criteria} on $f$ and Eq.~\eqref{eqn:QEI} for $g=f\varphi$.

The only issue that remains is estimating the null energy for $0 \leq \lambda \leq \ell_0$. This can be done with a pointwise estimate of the form
\be
\label{eqn:pointwise}
R_{\mu \nu} U^\mu U^\nu\big|_{\gamma(\lambda)} \geq \rho_0\,, \quad\quad \forall\lambda\in [0, \ell_0] \,,
\ee
where $\rho_0$ is a real constant that can be positive or negative. Even though this is a pointwise condition, the result is still stronger than having just the NEC. In this case, the $\rho_0$ is allowed to be negative and it is only assumed to hold for finite values of the affine parameter on the geodesic. 

We can write \(f^2 = (\varphi f)^2 + (1 - \varphi^2)\) and then we have
	\begin{align}
	\label{eqn:estim}
		\int_0^{\ell} f(\lambda)^2R_{\mu \nu} U^\mu U^\nu &\ge \int_0^{\ell_0} (1 - \varphi(\lambda)^2)(R_{\mu \nu} U^\mu U^\nu) d\lambda - Q_m(\gamma) \vert\vert (\varphi f)^{(m)}\vert\vert^2 - Q_0(\gamma) \vert\vert \varphi f\vert\vert^2 \\
		&\ge  \rho_0 \ell_0 -\rho_0 \|\varphi \|^2 - Q_m(\gamma) \vert\vert (g)^{(m)}\vert\vert^2 - Q_0(\gamma) \vert\vert g \vert\vert^2 \,,
	\end{align}
where we used Eq.~\eqref{eqn:QEI} and Eq.~\eqref{eqn:pointwise}. Then 
\be
\label{eqn:Jbound}
J_\ell[f] \leq -\rho_0 \ell_0 +\rho_0 \|\varphi \|^2 +(n-2)\|f'\|^2+ Q_m(\gamma) \vert\vert (g)^{(m)}\vert\vert^2 + Q_0(\gamma) \vert\vert g \vert\vert^2 \,.
\ee
Note that this expression no longer depends on the curvature, but the constants originate in the energy condition and the choice of functions $f$ and $\varphi$. One way to find appropriate functions $f$ and $\varphi$ is to solve the two variational problems: 
\be
\inf_{\varphi \in C^\infty[0, \ell_0)} \int_0^{\ell_0} \left(Q_m(\gamma) (\varphi^{(m)})^2+(\rho_0+Q_0(\gamma)) \varphi^2 \right) d\lambda \,,
\ee
for $\varphi$, and
\be
\inf_{f \in C^\infty(\ell_0, 
\ell]} \int_{\ell_0}^\ell \left(Q_m(\gamma) (f^{(m)})^2+(n-2)(f')^2+Q_0(\gamma) f^2 \right) d\lambda \,,
\ee
for $f$.  The boundary conditions are $\varphi^{(k)}(0)=0$, $\varphi (\ell_0)=1$, $f(\ell_0)=1$ and $f^{(k)}(\ell)=0$ for all $0\leq k \leq m-1$. The two variational problems are very complex for general number of derivatives $m$. In Ref.~\cite{fewster2020new} incomplete beta functions were chosen as their Sobolev norms can be computed in a closed form for general $m$. We also note that in two applications with null-integrated energy conditions, the maximum number of derivatives is $1$. In that case we can solve the variational problems exactly and find the optimal $f$ and $\phi$. We present both methods below. 

\begin{theorem}
\label{the:genm}
	Suppose that
	\begin{enumerate}
	    \item[(i)]
	 \((M, g_{\mu\nu})\) is a strongly asymptotically predictable spacetime; 
	 \item[(ii)]
	 for \(H\) the event horizon, \(U^{\mu}\) the tangent field of its null generators $\gamma(\lambda)$, and \(\mathrm{H}^{\mu}\) the mean normal curvature of \(\mathscr{H} = \Sigma \cap H\), where \(\Sigma\) is a Cauchy surface for the globally hyperbolic region \(\tilde{V}\) Eq.~\eqref{eqn:QEI} holds, and additionally there exist $0<\ell_0 \leq\ell$ and $\rho_0 \in \mathbb{R}$ such that 
	 \be
	 R_{\mu \nu} U^\mu U^\nu \big|_{\gamma(\lambda)} \geq \rho_0 \,.
	 \ee
	 \end{enumerate}
	 Then
	 \be
% \label{eqn:horizonch}
		\delta_U\mathcal{A}_{\mathscr{H}} = \int_{\mathscr{H}} \mathrm{H}^{\mu}(p)U_{\mu} \ge - \frac{1}{n - 2}\nu(Q_m(\gamma),Q_0(\gamma),\ell,\ell_0, \rho_0)\cdot\mathcal{A}_{\mathscr{H}} \,,
\ee
where 
\be
\nu(Q_m(\gamma),Q_0(\gamma),\ell,\ell_0, \rho_0)=-(1-A_m)\rho_0\ell_0 +\frac{Q_mC_m}{\ell_0^{2m-1}} + Q_0A_m\ell + \frac{(n - 2)B_m}{\ell - \ell_0} + \frac{Q_mC_m}{(\ell-\ell_0)^{2m-1}} \,,
\ee
and
\be
\label{eqn:constants}
	A_m = \frac{1}{2} - \frac{(2m)!^4}{4(4m)!m!^4} \,, \quad 
	B_m= \frac{(2m-2)!^2(2m-1)!^2}{(4m-3)!(m - 1)!}\,, \quad 
	C_m = \frac{(2m-2)!(2m-1)!}{(m-1)!^2} \,.
\ee
\end{theorem}

\begin{proof}
The first part of the proof is similar to Lemma 4.5 of Ref.~\cite{fewster2020new}. We pick functions 
\be
	f(\lambda) = 
		\begin{cases}
			1 \hfill \lambda\in [0, \ell_0) \\
			I(m, m; \frac{\ell - \lambda}{\ell - \ell_0}) \quad \hfill \lambda\in [\ell_0, \ell]\,,
		\end{cases}
\ee
and
\be
	\varphi(\lambda) = 
	\begin{cases}
		I(m, m;\frac{\lambda}{\ell_0}) \quad \hfill \lambda \in [0, \ell_0) \\
		1 \hfill \lambda \in [\ell_0, \ell] \,,
	\end{cases}	
\ee
noting that they satisfy the requirements explained above. Then $J_\ell[f]$ obeys Eq.~\eqref{eqn:Jbound}. From Appendix A of Ref.~\cite{fewster2020new} we have that
\be
	\vert\vert \varphi f\vert\vert^2 = A_m\ell_0 + A_m(\ell - \ell_0)\quad\quad
		\vert\vert (\varphi f)^{(m)}\vert\vert^2 = \frac{C_m}{\ell_0^{2m - 1}} + \frac{C_m}{(\ell - \ell_0)^{2m - 1}}
		\quad\quad 
		\vert\vert f'\vert\vert^2 = \frac{B_m}{\ell - \ell_0} \,,
\ee
where the constants are given by Eq.~\eqref{eqn:constants}. Then $J_\ell[f] \leq \nu(Q_m(\gamma),Q_0(\gamma),\ell,\ell_0, \rho_0)$. The rest of the proof follows Theorem~\ref{th:general} concluding via contradiction that $H^\mu (p)U_\mu \geq -(1/(n-2))\nu(Q_m(\gamma),Q_0(\gamma),\ell,\ell_0, \rho_0)$.
\end{proof}

Now we turn to the $m=1$ case. We can prove the following theorem

\begin{theorem}
\label{the:m=1}
Suppose that
	\begin{enumerate}
	    \item[(i)]
	 \((M, g_{\mu\nu})\) is a strongly asymptotically predictable spacetime; 
	 \item[(ii)]
	 for \(H\) the event horizon, \(U^{\mu}\) the tangent field of its null generators $\gamma(\lambda)$, and \(\mathrm{H}^{\mu}\) the mean normal curvature of \(\mathscr{H} = \Sigma \cap H\), where \(\Sigma\) is a Cauchy surface for the globally hyperbolic region \(\tilde{V}\), 
  \be
  \label{eq:QEI}
\int_0^{\ell} g(\lambda)^2 R_{\mu\nu}U^{\mu}U^{\nu} \ge -Q_1(\gamma) \vert\vert g'\vert\vert^2 - Q_0(\gamma) \vert\vert g\vert\vert^2 \,.
\ee

  holds, and additionally there exist $0<\ell_0 \leq\ell$ and $\rho_0 \in \mathbb{R}$ such that 
	 \be
    \label{eq:rho0}
	 R_{\mu \nu} U^\mu U^\nu \big|_{\gamma(\lambda)} \geq \rho_0 \qquad \text{for} \qquad \lambda \in[0,l_0] \,.
	 \ee
	 \end{enumerate}
	 Then
	 \be
		\delta_U\mathcal{A}_{\mathscr{H}} = \int_{\mathscr{H}} \mathrm{H}^{\mu}(p)U_{\mu} \ge - \frac{1}{n - 2}\nu_1(Q_1(\gamma),Q_0(\gamma),\ell,\ell_0, \rho_0)\cdot\mathcal{A}_{\mathscr{H}} \,,
\ee
where
\bea
\label{eqn:nuonebound}
\nu_1(Q_1(\gamma),Q_0(\gamma),\ell,\ell_0, \rho_0)&=&-\rho_0 \ell_0+\sqrt{Q_0(\gamma)(Q_1(\gamma)+n-2)} \coth{\left(\frac{(\ell-\ell_0)\sqrt{Q_0(\gamma)}}{\sqrt{Q_1(\gamma)+n-2}}\right)} \nonumber \\
&&+\sqrt{Q_1(\gamma) (Q_0(\gamma)+\rho_0)} \coth{\left(\frac{\ell_0\sqrt{Q_0(\gamma)+\rho_0}}{\sqrt{Q_1(\gamma)}}\right)} \,.
\eea

\end{theorem}
\begin{proof}
From Eq.~\eqref{eqn:Jbound} and assumption (i) we have
\bea
J_\ell[f]&\leq& -\rho_0 \ell_0 +\inf_{\varphi \in C^\infty[0, \ell_0)} \int_0^{\ell_0} \left(Q_1(\gamma) (\varphi')^2+(\rho_0+Q_0(\gamma)) \varphi^2 \right) d\lambda \nonumber\\
&&+\inf_{f \in C^\infty(\ell_0, 
\ell]} \int_{\ell_0}^\ell \left((Q_1(\gamma)+n-2)(f')^2+Q_0(\gamma) f^2 \right) d\lambda \,.
\eea
The solutions to the two variational problems with boundary conditions $f(\ell_0)=1$, $f(\ell)=0$ and $\varphi(0)=0$, $\varphi(\ell_0)=1$ give the functions
\be
\label{eqn:f}
	f(\lambda) = 
		\begin{cases}
			1 \hfill \lambda\in [0, \ell_0) \\
			 \csch{\left[\alpha (\ell-\ell_0) \right]} \sinh{\left[\alpha (\ell-\lambda) \right]} \quad \hfill \lambda\in [\ell_0, \ell]\,
		\end{cases}
  \quad  \alpha=\sqrt{\frac{Q_0(\gamma)}{Q_1(\gamma)+n-2}}\,,
\ee
and
\be
\label{eqn:varphi}
	\varphi(\lambda) = 
	\begin{cases}
		\csch{\left(\beta \ell_0 \right)} \sinh{\left(\beta \lambda \right)} \quad \hfill \lambda \in [0, \ell_0) \\
		1 \hfill \lambda \in [\ell_0, \ell] \,,
	\end{cases}	
 \quad \beta=\sqrt{\frac{Q_0(\gamma)+\rho_0}{Q_1(\gamma)}}
\ee
where in the special case \(Q_0(\gamma) + \rho_0 = 0\)
\be
\label{eqn:varphi-0}
	\varphi(\lambda) = 
	\begin{cases}
		\frac{1}{\ell_0} \quad \hfill \lambda \in [0, \ell_0) \\
		1 \hfill \lambda \in [\ell_0, \ell] \,,
	\end{cases}	
\ee

\begin{figure}
\includegraphics[scale=0.7]{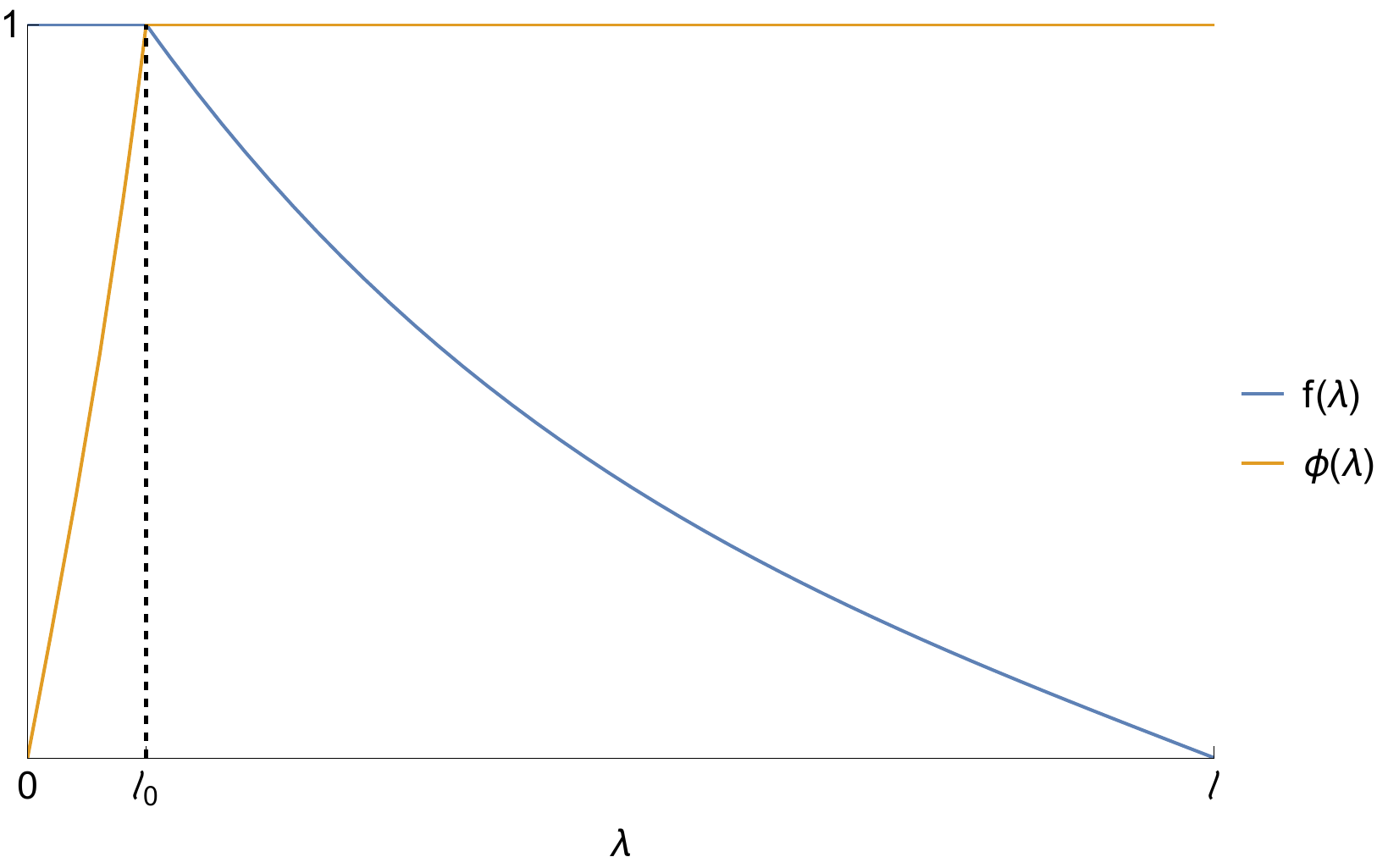}
\caption{The functions $f(\lambda)$ and $\varphi(\lambda)$ given in Eq.~\eqref{eqn:f} and \eqref{eqn:varphi} for specific values of the parameters.}
\label{fig:functions}
\end{figure}

A plot of the functions is given in Fig.~\ref{fig:functions}. Whenever $Q_0(\gamma)+\rho_0 < 0$ the hyperbolic functions in the definition of \(\varphi\) become trigonometric functions. 

With these functions we have $J_\ell[f]\leq \nu_1(Q_1(\gamma),Q_0(\gamma),\ell,\ell_0,\rho_0)$. The rest of the proof follows Theorems~\ref{the:genm} and \ref{th:general}.
It is worth pointing out how in the limit \(Q_0 + \rho_0 \rightarrow 0\) the upper bound reduces to:
\be
\nu_1(Q_1(\gamma),Q_0(\gamma),\ell,\ell_0, \rho_0)=-\rho_0 \ell_0+\sqrt{Q_0(\gamma)(Q_1(\gamma)+n-2)} \coth{\left(\frac{(\ell-\ell_0)\sqrt{Q_0(\gamma)}}{\sqrt{Q_1(\gamma)+n-2}}\right)} + \frac{Q_1(\gamma)}{\ell_0} \,,  
 \ee
 so no divergence appears, and additionally \(\nu_1\) is continuous, as expected.

\end{proof}
		
\section{Applications}
\label{sec:applications}

In this section we will use two models of field theories, one classical and one quantum, that violate the NEC as well as the more general condition of Theorem~\ref{the:genclas}, thus allowing for black hole evaporation. For those two models we will compare the bound on the evaporation rate given by Theorem~\ref{the:m=1} with the explicit evaporation rate calculated for spherically symmetric spacetimes. Throughout the section we slightly simplify the notation, implying for the coefficients that $Q_i=Q_i(\gamma)$.

The evaporation rate of a black hole, the rate at which its mass is decreasing due to Hawking radiation, was computed early on for non-rotating black holes by Page \cite{page1976particle}
\be
\frac{1}{M} \frac{dM}{dt}=-\frac{\hbar c^4 \alpha}{G^2}\frac{1}{M^3} \,,
\ee
where $\alpha$ is a numerical coefficient that depends on the particle species emitted. In terms of the Hawking temperature
\be
\label{eqn:evaprate}
\nu_{\text{ev}}=-\frac{1}{M} \frac{dM}{dt}=(8\pi)^3 \alpha \left(\frac{k}{T_{\text{pl}}^2\hbar}\right) T^3 \,.
\ee
For a spherically symmetric black hole, the evaporation rate coincides with the rate of decrease of the black hole horizon up to a multiplicative constant
\be
\frac{\delta_U\mathcal{A}_{\mathscr{H}}}{\mathcal{A}_{\mathscr{H}}} = (n - 1) \frac{\delta_U{M}}{M} = -(n - 1)\nu_{\text{ev}}.
\ee

\subsection{Optimization}

Before discussing any of the examples, we will examine the dependence of $\nu_1$ on the affine length parameters $\ell$ and $\ell_0$. It is rather immediate to see that \(\nu_1\) in~\eqref{eqn:nuonebound} realizes its minimum for $\ell\to \infty$ as the only $\ell$−dependence is in the hyperbolic cotangent term, which is monotonically decreasing towards $1$. Unlike the singularity theorems case, where a finite $\ell$ was giving information about the location of the singularity, here we can take $\ell\to \infty$ and simplify the derived bound. Beyond this intuition, it makes physical sense to require an energy condition of the form \eqref{eq:QEI} to be obeyed for \(\ell \rightarrow +\infty\); the geodesics considered all extend to an infinite value of the affine parameter, due to the causality condition, and there is no reason to expect we should restrict to looking at only a sub-portion of it. In other words, taking \(\ell \rightarrow +\infty\) should be seen as equivalent to taking into account as much information as we have access to, and hence retrieve the best possible bound.

Turning to $\ell_0$, the matter is a bit more subtle. We will divide the discussion in three cases:\\

\noindent
(i) $\rho_0 > 0$ \\

\noindent
In this case we can always find a finite $\ell_0$ that makes $\nu_1=0$. In particular the value of $\ell_0$ is the solution of
\be
-|\rho_0| \ell_0+\sqrt{Q_0(Q_1+n-2)}+\sqrt{Q_1 |Q_0+\rho_0|} \coth{\left(\frac{\ell_0\sqrt{|Q_0+\rho_0|}}{\sqrt{Q_1}}\right)}=0 \,.
\ee
This is the case where the original Hawking area theorem applies and evaporation is prohibited. However, if the problem provides a value of $\ell_0$ smaller than this critical value then evaporation is allowed and an evaporation bound can be computed.\\

\noindent
(ii) $\rho_0 \leq 0$ and $\rho_0 + Q_0 \ge 0$ \\

\noindent
The minimum of $\nu_1$ is found when $\ell_0=\tilde{\ell}_0$ where
\be
\label{eq:l0hyperbolic}
\tilde{\ell}_0=\sqrt{\frac{Q_1}{|Q_0+\rho_0|}}\arcsinh{\left(\sqrt{\frac{|Q_0+\rho_0|}{|\rho_0|}}\right)} \,,
\ee
which in this case is a global minimum. The case of $\rho_0=0$ gives an asymptotic minimum value of $\nu_1$ equal to $\sqrt{Q_0(Q_1+n-2)}+\sqrt{Q_1 Q_0}$ for $\ell_0 \to \infty$. \\

\noindent
(iii) \(\rho_0< 0\) and \(\rho_0 + Q_0 < 0\) \\

\noindent

In this case, the hyperbolic cotangent of Eq.~\eqref{eqn:nuonebound} turns into a trigonometric one, and there always exist a finite value of \(\ell_0\) which makes \(\nu_1\) divergent, namely
\be
\ell_0 = \pi\sqrt{\frac{Q_1}{\vert Q_0 + \rho_0\vert}} \,.
\ee
However, this value should be excluded. The reason is, it lies in a regime where the pointwise energy condition $R_{\mu \nu}U^\mu U^\nu \geq -|\rho_0|$ contradicts the averaged energy condition of Eq.~\eqref{eq:QEI} as we analyze in Appendix~\ref{app:l0upperbound}. In particular we find that the upper bound for $\ell_0$ in this case is 
\be
\label{eqn:ellbound}
\ell_0 \leq \frac{\pi}{2} \sqrt{\frac{Q_1}{\vert Q_0 + \rho_0\vert}},
\ee
which not only excludes all the divergent values of \(\nu_1\), but also all its negative values, consistently proving that this set of energy conditions doesn't forbid black hole evaporation. The only minimum of $\nu_1$ in the allowed region is
\begin{equation}
    \label{eq:l0-trigonometric}
    \tilde{\ell_0}=\sqrt{\frac{Q_1}{|\rho_0 + Q_0|}}\arcsin{\left(\sqrt{\frac{|\rho_0 + Q_0|}{|\rho_0|}}\right)}\,, 
\end{equation}
which is the analytic continuation of the expression for \(\tilde{\ell}_0\) given in \eqref{eq:l0hyperbolic}. Thus \(\tilde{\ell_0}\) is continuous for all values of $\rho_0+Q_0$. The $\tilde{\ell_0}$ is always less than the upper bound of Eq.~\eqref{eqn:ellbound} as the function $\arcsin{x}$ has a maximum value of $\pi/2$. The three cases are visualized in Fig.~\ref{fig:cases}.

\begin{figure}
     \centering
     \begin{subfigure}[b]{0.3\textwidth}
         \centering
   \includegraphics[height=3.4cm]{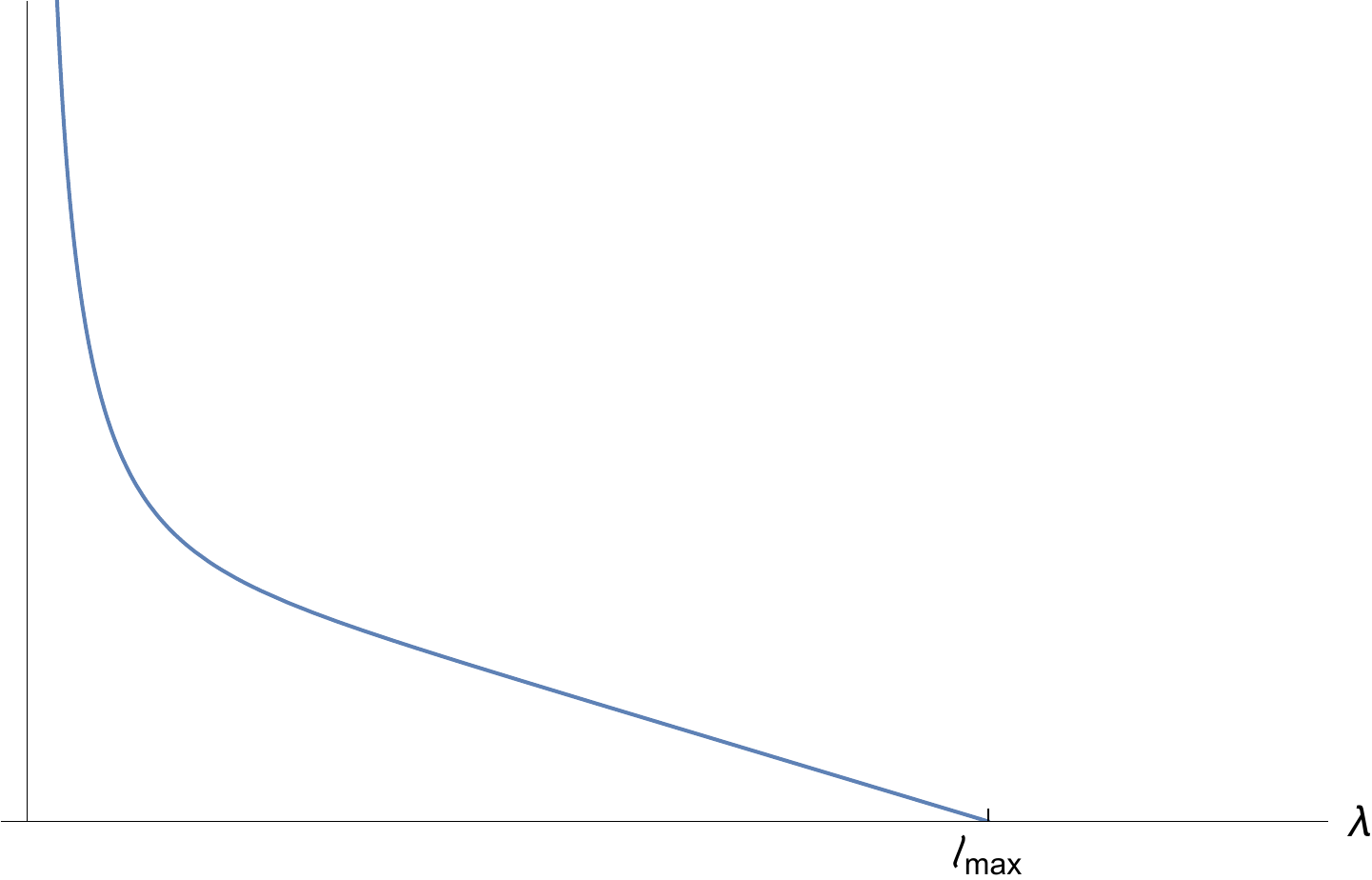}
         \caption[caption]{Case (i)  \newline $\!\!\!\rho_0 > 0$}
         \label{fig:(i)}
     \end{subfigure}
     \hfill
     \begin{subfigure}[b]{0.3\textwidth}
         \centering
   \includegraphics[height=3.4cm]{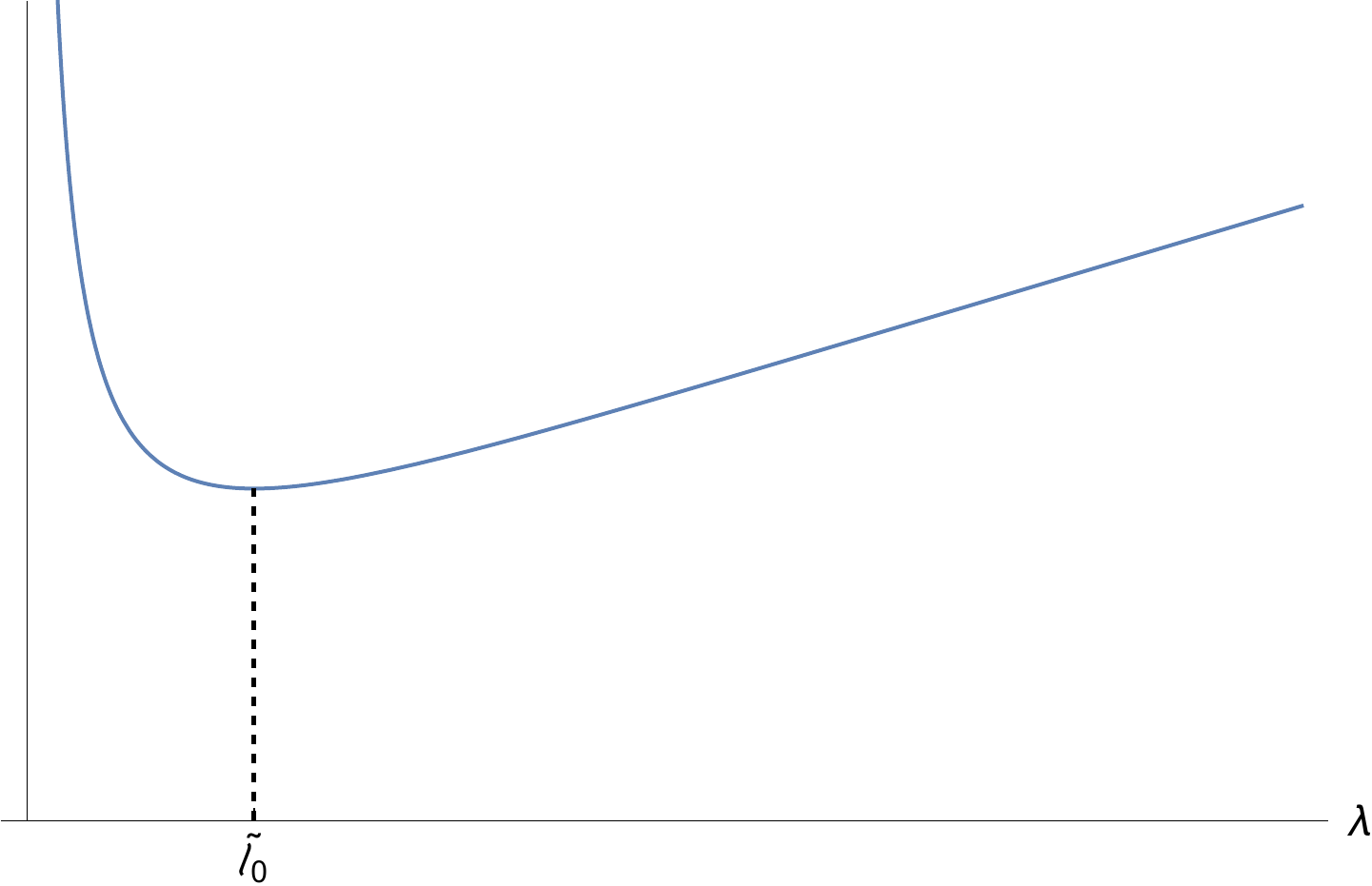}
         \caption[caption]{Case (ii)\newline
         $\rho_0 \leq 0$ and $Q_0+\rho_0 \geq 0$}
         \label{fig:(ii)}
     \end{subfigure}
     \hfill
     \begin{subfigure}[b]{0.3\textwidth}
         \centering
\includegraphics[height=3.4cm]{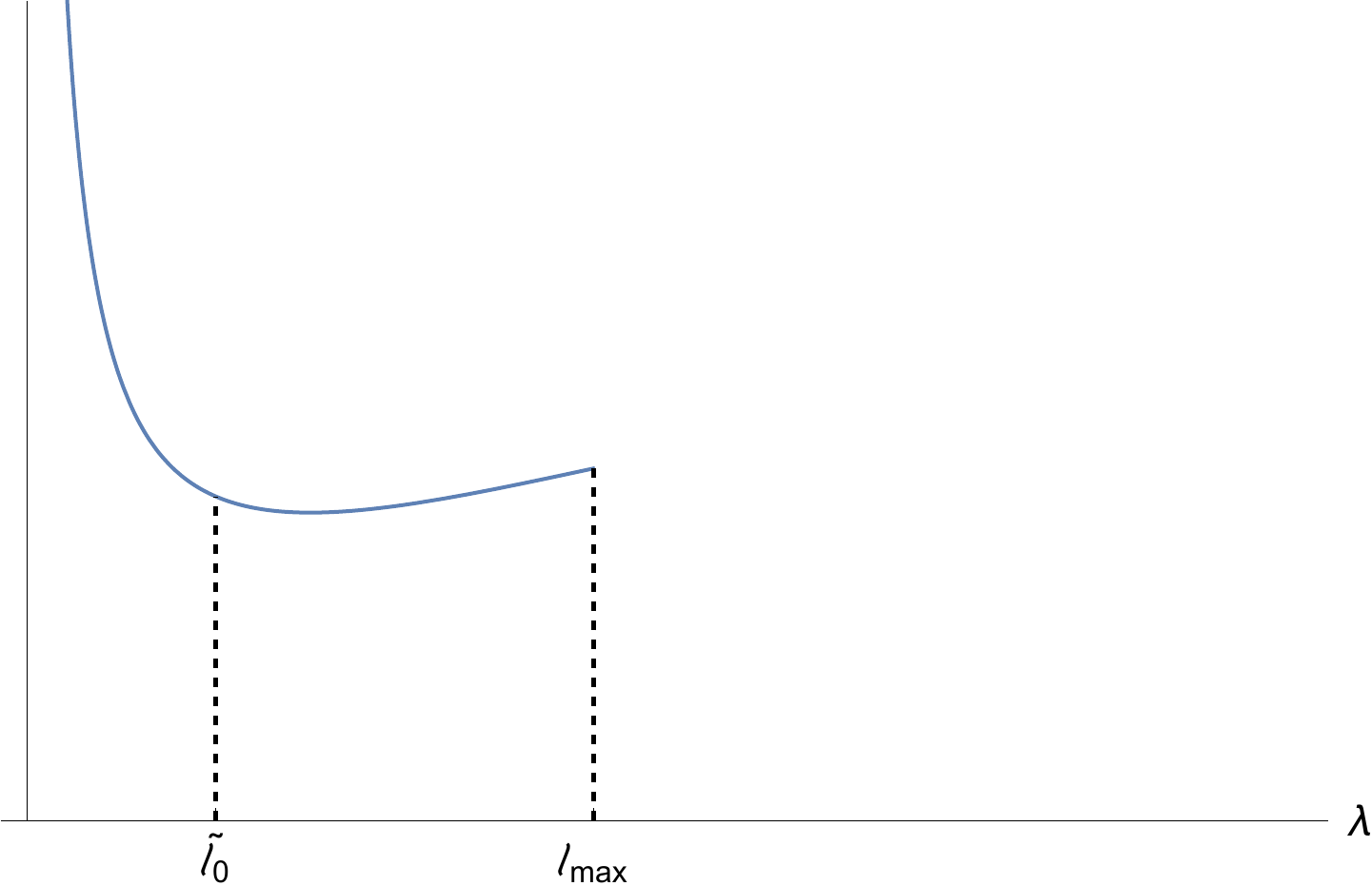}
         \caption[caption]{Case (iii)\newline
         $\rho_0 < 0$ and $Q_0+\rho_0 < 0$}
         \label{fig:(iii)}
     \end{subfigure}
        \caption{Sample plots of $\nu_1$ for the three cases of constant values. In graph (a) the $\ell_{\max}$ represents the value $\ell_0$ for which the evaporation rate becomes zero and the original theorem applies. In graph (ii) there is a global minimum, $\tilde{\ell_0}$ given in Eq.~\eqref{eq:l0hyperbolic}. In graph (iii) there is a maximum value of $\ell_0$ after which the QEI is not obeyed (see Appendix~\ref{app:l0upperbound}). In the allowed range there is a local minimum $\tilde{\ell}_0$ given in Eq.~\eqref{eq:l0-trigonometric}.}
        \label{fig:cases}
\end{figure}

From now on we will assume that $\rho_0<0$ which is the case of interest for black hole evaporation, although, as we discussed, our method is also applicable for non-negative $\rho_0$. The optimized evaporation bound for four-dimensions ($n=4$) reduces to 
\bea
\label{eqn:nuopt}
\nu_{\text{opt}}(Q_0,Q_1,\rho_0)&=&\sqrt{Q_0(2+Q_1)}+\sqrt{Q_1 Q_0} \nonumber \\
&&+|\rho_0| \sqrt{\frac{Q_1}{Q_0 +\rho_0}}\arcsinh{\left(\sqrt{\frac{Q_0 + \rho_0}{|\rho_0|}}\right)} \,,
\eea
which includes both cases (ii) and (iii). The constants $Q_0$ and $Q_1$ will include the details of the specific matter model we use via the relevant energy condition. The pointwise bound $\rho_0$ is generally undetermined as quantum fields have no pointwise lower bounds for all states. However, for applications, we can get a reasonable value for it considering estimations of the null energy density around the black hole horizon. For that purpose we use the numerical results of Levi and Ori \cite{levi2016versatile} for the renormalized stress-energy tensor of a minimally coupled scalar field near the horizon of a Schwarzschild black hole. They find the value of the null energy on the $r=3M$ geodesic to be
\be
\label{eqn:rhozero}
\rho_0 \approx -2.7 \times 10^{-7} \frac{\hbar c^9}{G^4 M^4} \,,
\ee
in SI units. The $M^{-4}$ dependence preceded by a negative coefficient -- the most relevant features of \eqref{eqn:rhozero} -- can also be found using the semi-analytical model of Visser~\cite{PhysRevD.56.936} for conformally coupled quantum scalar fields\footnote{The stress tensor averaged over the Unruh vacuum is analytically constrained up to some coefficients that Ref.~\cite{PhysRevD.56.936} determines numerically; relying on the same framework, but changing to Kruskal coordinates, we have been able to compute the contraction \(T_{\mu\nu}U^{\mu}U^{\nu}\) on the horizon of a Schwarzschild black hole, in presence of scalar fields conformally coupled to the metric. This process was also described in \cite{christensen1977sa}.}. Thus we will proceed to bound the pointwise null energy near the black hole horizon using Eq.\eqref{eqn:rhozero}. 

Using the expression for Hawking temperature 
\be
T_H=\frac{\hbar c^3}{8\pi G M k} \,,
\ee
where $k$ is the Boltzmann constant, we can write $\rho_0$ in terms of $T_H$ 
\be
\label{eqn:rhozerofin}
\rho_0\approx- 0.1 \frac{k^4}{\hbar^3 c^3} T^4 \,.
\ee

\subsection{The non-minimally coupled classical Einstein--Klein--Gordon theory}
\label{subsec:non-min-EKG-theory}

Non-minimally coupled scalar fields is the typical classical example that violate the NEC \cite{kontou2020energy}, and thus the original black hole area theorem doesn't apply. These scalar fields are described by the Lagrangian density
\be
\mathcal{L}[\phi] = -\frac{1}{2}\left[\nabla_{\mu}\phi\nabla^{\mu}\phi   - (m^2 - \xi R)\phi^2\right] \,,
\ee
where $\xi$ is the coupling constant and $m$ the mass of the field. The corresponding stress-energy tensor acquired by varying the action is
\begin{equation}
    T_{\mu\nu} = \nabla_{\mu}\phi\nabla_{\nu}\phi - \frac{1}{2}g_{\mu\nu}\left[\nabla_{\rho}\phi\nabla^{\rho}\phi - m^2\phi^2\right] + \xi\left(g_{\mu\nu}\square_g - \nabla_{\mu}\nabla_{\nu} + G_{\mu\nu}\right)\phi^2 \,.
\end{equation}
To proceed, we assume that the non-minimally coupled scalar field obeys the bound $\phi^2 < 1/8\pi\xi$. This is a reasonable assumption as if the field value is larger, the effective Newton constant changes sign \cite{kontou2020energy}. Furthermore, we assume that the coupling constant is $\xi \in [0,1/4]$, where $1/4$ is always smaller that the conformal coupling $\xi_c= (n - 2)/(4(n - 1))$. Then we can state the following bound proven in Ref.~\cite{brown2018singularity} (and applying semiclassical Einstein Equations):
\be
        \int_{\gamma}g^2 R_{\mu\nu}U^{\mu}U^{\nu} \ge -Q\left(\|g'\|^2 + \tilde{Q}^2 \vert\vert g\vert\vert^2\right) \,,
\ee
    where 
\be
    Q = \frac{32\pi\xi\phi_{\max}^2}{1 - 8\pi\xi\phi_{\max}
    ^2}\,,
    \quad\quad
    \tilde{Q} = \frac{8\pi\xi\phi_{\max}\phi'_{\max}}{1 - 8\pi\xi\phi_{\max}^2} \,.
\ee
and
\be
\phi_{\max}=\sup_{\gamma}|\phi|\,, \quad \phi'_{\max}=\sup_{\gamma}\vert \phi'(\lambda) \vert \,.
\ee
Then this is a bound of the form of Eq.~\eqref{eqn:QEI} with $m=1$, $Q_0=Q\tilde{Q}^2$ and $Q_1=Q$.  

To proceed, we need to estimate the values of $\phi_{\max}$ and $\phi'_{\max}$. In principle the field values may be not bounded, however it is reasonable to connect the scale of the field magnitude with a temperature. In order to accomplice that we follow the hybrid approach of Refs.~\cite{brown2018singularity,fewster2020new}. We take for $\phi_{\max}$ as the value of the Wick square $\langle \nord{\phi^2} \rangle_\omega$ in Minkowski spacetime, where \(\omega\) is some state of the quantum field theory. Then we specify \(\omega\) to be a thermal equilibrium KMS state \cite{haag2012local}, and connect $\phi_{\max}$ to a temperature $T$. The result for massless fields was derived in Ref.~\cite{brown2018singularity} 
\be
\phi_{\max}^2 \sim \langle \nord{\phi^2} \rangle_T=\lim_{x'\rightarrow x} \left[W_T(x, x') - W_0(x, x')\right]=\frac{T^{n - 2}}{2^{n - 2}\pi^{\frac{n - 1}{2}}}\frac{\Gamma(n - 2)}{\Gamma\left(\frac{n - 1}{2}\right)}\zeta(n - 2) \,.
\ee
Here $W_T$ is the two point function of the state with temperature $T$ and $\zeta$ is the Riemann zeta function. Similarly for the $\phi'_{\max}$ and $U^\mu$ any null vector with $U^0=1$ we have
\be
(\phi'_{\max})^2 \sim \langle \nord{(U^\mu \nabla_\mu \phi) (U^\nu \nabla_\nu \phi) } \rangle_T=\frac{nT^n}{2^{n - 1}\pi^{\frac{n - 1}{2}}} \frac{\Gamma(n)}{\Gamma\left(\frac{n + 1}{2}\right)} \zeta(n) \,.
\ee
Substituting $n=4$ we have
\be
\phi_{\max}^2\sim \frac{T^2}{12} \,, \qquad (\phi'_{\max})^2\sim \frac{2\pi^2 T^4}{45} \,,
\ee
and the corresponding coefficients 
\be
Q_1=\frac{8\pi \xi (T/T_{\text{pl}})^2}{3(1-(2/3)\pi \xi(T/T_{\text{pl}})^2)} \,, \qquad Q_0=\frac{256 \pi^5 \xi^3 (T^8/T_{\text{pl}}^6)}{405(1-(2/3)\pi \xi(T/T_{\text{pl}})^2)} \left(\frac{k}{\hbar}\right)^2 \,,
\ee
where we have restored the units. Here  $T_{\text{pl}}$ is the Planck temperature and we note that $Q_1$ is dimensionless while $Q_0$ has dimensions of $s^{-2}$.

For the $\rho_0$ we will use the value of Eq.~\eqref{eqn:rhozerofin}. Even though that calculation was conducted within a semi-classical framework, here we only aim for an estimate for a pointwise value of the null energy at the black hole horizon. 

With these considerations we can explicitly express the bound of Eq.~\eqref{eqn:nuopt}, in terms of the background temperature. Consistently with the semi-classical approximation, we assume the temperature of the black hole to be far below the Planck scale, and so we expand $\nu$ in terms of $(T/T_{\text{pl}})$, obtaining at leading order:
\be
\nu_{\text{opt}}(T,\xi)=\sqrt{\frac{\pi^3}{15}} \sqrt{\xi} \left(\frac{k}{\hbar T_{\text{pl}}^2} \right)T^3+\mathcal{O}(T/T_{\text{pl}})^4 \,.
\ee
Interestingly this temperature dependence is the same as that of the evaporation rate for spherically symmetric black holes computed in Eq.~\eqref{eqn:evaprate}. We note that this result is not given trivially by dimensional analysis, as $T/T_{\text{pl}}$ is dimensionless. 
As usual, this can also be recasted into the typical semiclassical expansion, by expressing \(T_{pl}\) in terms of $G$. 

The numerical coefficients can allow us to explicit how strong are the constraints on the evaporation rate. The first observation is that the larger the coupling constant $\xi$ the larger the allowed evaporation rate. This is expected as classical fields with \(\xi = 0\) respect Eq.~\eqref{eqn:classcond}, and so the classical picture applies. The maximum value of $\xi$ for which $\nu_{\text{opt}} \leq \nu_{\text{ev}}$ is 
\be
\xi_{\max}=1.2 \times 10^8 \alpha^2 \,.
\ee
For large black holes $\alpha \sim 2\times 10^{-4}$ as shown in Ref.~\cite{page1976particle} so $\xi_{\max} \sim 4.8$. Thus values of $\xi \in [0,\xi_c]$ provide an additional bound to the evaporation rate. Of course we should note that this largely depends on the assumptions made for the pointwise energy condition and $\rho_0$. Additionally, the energy condition is for a classical field, only providing an analogy of how this condition could look for a quantum non-minimally coupled scalar. So this result merely showcases how our theorem can be used in principle to constrain the black hole evaporation rate for certain models. 

\subsection{The smeared null energy condition}

In this example we will use the smeared null energy condition (SNEC) of Eq.~\eqref{eqn:SNEC} as our energy assumption. As mentioned in the introduction, to use that condition for classical relativity theorems we need to convert it to a curvature condition using the semiclassical Einstein Equation of Eq.~\eqref{eqn:see}
\be
        \int_{\gamma}g^2(\lambda) R_{\mu\nu}U^{\mu}U^{\nu} \ge -32\pi B \|g'\|^2  \,.
\ee
This is an equation of the form of Eq.~\eqref{eqn:QEI} with $Q_1=32 \pi B$ and $Q_0=0$.

This equation has been proven for minimally coupled quantum scalar fields only for Minkowski spacetime \cite{Fliss:2021gdz} and the coefficient $Q_1$ depends on the undetermined constant $B$. Its value depends on the cutoff scale $\luv$ as when
\be
\label{eqn:GNUV}
N G \lesssim \luv^{d-2} \,.
\ee 
is saturated, meaning $\luv$ is the Planck length scale, $B=1/32\pi$. That was found in the induced gravity proof of \cite{leichenauer2019upper}. When \eqref{eqn:GNUV} is not saturated, such as in controlled effective theory constructions, we have $B\ll 1$. In these cases $\luv$ is far above the Planck length scale. This issue is also discussed in detail in Ref.~\cite{Freivogel:2020hiz}. 

Using these values for $Q_1$ and $Q_0$ as well as the estimated value of $\rho_0$ from Eq.~\eqref{eqn:rhozerofin} we have for $\nu_{\text{opt}}$ of Eq.~\eqref{eqn:nuopt}
\be
\nu_{\text{opt}}(B,T)=\frac{2\sqrt{2\pi^3 }}{\sqrt{5}} \sqrt{B} \left(\frac{k}{\hbar T_{\text{pl}}}\right) T^2 \,.
\ee

We note two things about that expression: First the smaller the constant $B$ the stricter the bound on the evaporation rate. That makes sense as the further away the UV cutoff is from the Planck scale, the more ``classical'' the theory is. Second, the temperature dependence is $T^2$, different from the estimated evaporation rate of Eq.~\eqref{eqn:evaprate} where the dependence is $T^3$. That means the evaporation rate from our theorem can provide a stricter bound only at higher temperatures. The largest possible $B$ where we have a stricter bound is found by setting the temperature equal to the Planck temperature and equating the two evaporation rates. The result is
\be
B_{\max}=1.63 \times 10^5 \pi^3 \alpha^2 \,.
\ee
For large black holes $\alpha \sim 2\times 10^{-4}$ and $B_{\max}=0.20$. This value is larger than the maximum allowed value for $B$ of $1/32\pi$. This means for a variety of effective field theories the derived evaporation bound can provide restrictions. The result for $B=1/32\pi$ is shown in Fig.~\ref{fig:SNEC}. 

\begin{figure}
\includegraphics[height=7cm]{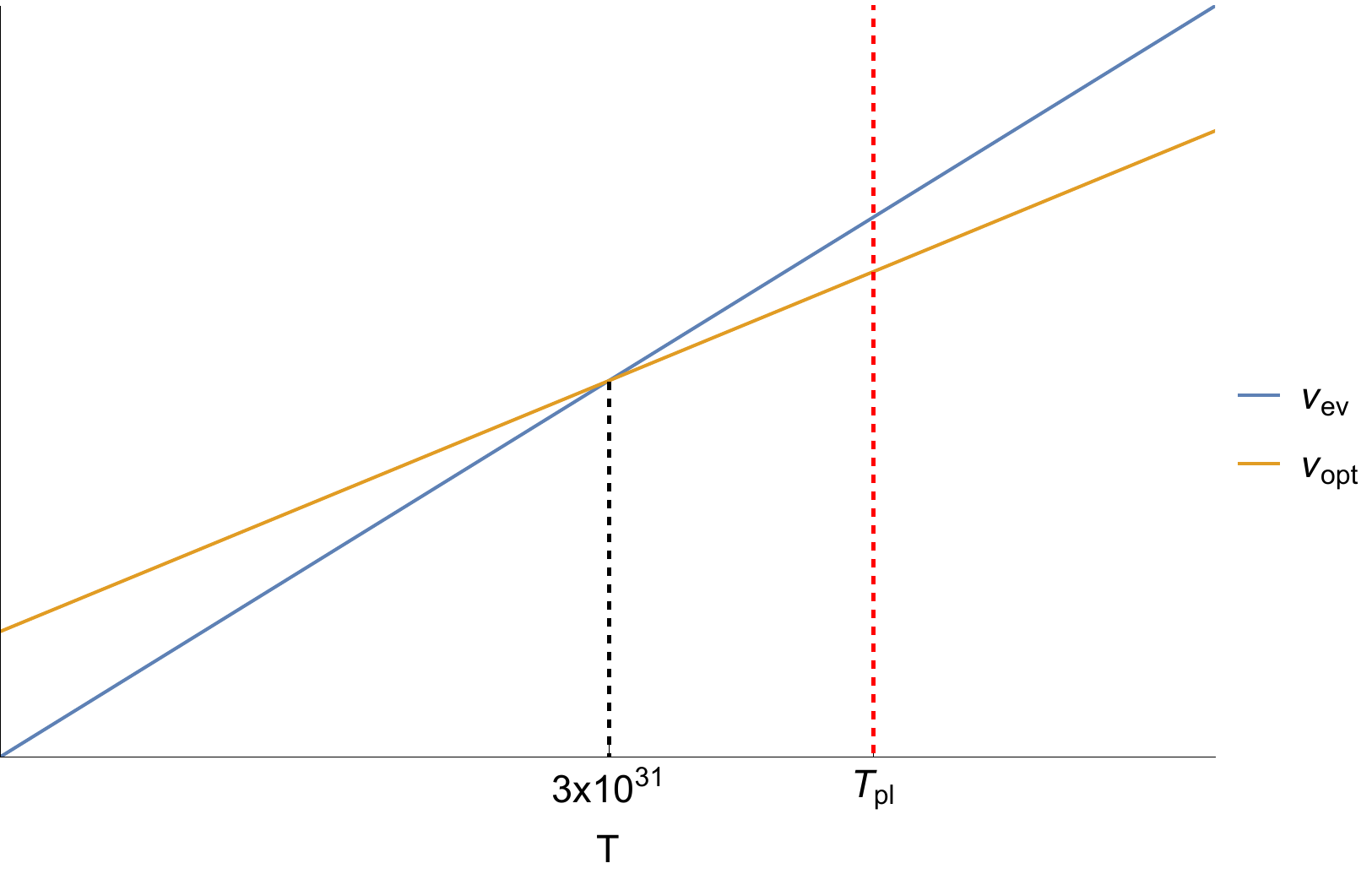}
\caption{The evaporation rate for spherical black holes (blue) given in Eq.~\eqref{eqn:evaprate} and the bound on the evaporation rate from our theorem using SNEC as the energy condition. Here we set the $B$ as its maximum value of $1/32\pi$.}
\label{fig:SNEC}
\end{figure}

\section{Discussion}
\label{sec:discussion}

In this work we derived two generalizations of the Hawking black hole area theorem: First we provided a more general condition than the null energy condition for which the conclusion of the original theorem still holds. Second for conditions inspired by quantum energy inequalities (QEIs) we proved a more general version of the theorem that, instead of prohibiting black hole evaporation, provides a bound on the black hole evaporation rate. For the last case we applied our theorem to two field theories that allow for shrinking of the black hole horizon: the classical non-minimally coupled scalar and the quantum minimally coupled scalar for a theory with a UV cutoff. While these are both toy models due to the classical nature of the field (in the first case) and lack of curvature corrections (in the second case) they provide an example of how our theorem could provide a meaningfully bound to the semiclassical black hole evaporation.

A significant obstacle to a realistic application in quantum field theory is the lack of a general finite bound on the energy integrated on a finite null segment, as discussed in the introduction. The smeared null energy condition bound used in our applications suffers from the need of a theory dependent UV cutoff. Additionally the only existing proof is for Minkowski spacetimes \cite{Fliss:2021gdz}. A better bound is that of the double smeared null energy condition, where the null energy is averaged over the two null directions \cite{Fliss:2021phs}. Unfortunately, it is unclear how such an inequality could be applied to a proof of the generalized area theorem. However, a recent derivation in the timelike case showed how a worldvolume inequality can be used in the proof of a singularity theorem \cite{Graf:2022mko} using segment inequalities and area comparison results. This work points to a way forward in the null case.

A different extension of this work would be to compare our results with the generalized second law of black hole thermodynamics. It has been shown that the second law can be used to derive the averaged null energy condition \cite{Wall:2009wi}. A connection of the evaporation rate bound we derived with the entanglement entropy of the Hawking radiation could provide meaningful insight to the semiclassical evaporation process.

Finally we should stress out that our result does only constraint the evaporation rate and it does not provide a mechanism for evaporation. More importantly the black holes with matter that violate the condition \eqref{eqn:classcond} for the classical area theorem are merely allowed to evaporate but do not necessarily do so. For example there is no known mechanism for classical fields to produce Hawking radiation, not even for the ones that violate the condition \eqref{eqn:classcond}. A recent work on semiclassical black hole evaporation \cite{Meda:2021zdw} showed that it is induced by a QEI-type condition which remains generally unproven. A connection of this work with our result could provide new insights.

\begin{acknowledgments}
E-AK would like to thank Chris Fewster and Paolo Meda for useful discussions. E-AK was partly supported by the ERC Consolidator Grant QUANTIVIOL. This work is part of the $\Delta$-ITP consortium, a program of the NWO that is funded by the Dutch Ministry of Education, Culture and Science (OCW).

VS would like to thank Luca Arnaboldi, Gimmy Tomaselli and Enrico Trincherini for useful discussions.
Part of this work was conducted as research for the Masters program in the University of Pisa, and was widely supported by the Scuola Normale Superiore of Pisa. Furthermore, VS would like to thank the University of Amsterdam for hosting her for two months while this research project was conducted. 

\end{acknowledgments}

\newpage

\appendix\numberwithin{equation}{section}

\section{Location of the trapped surface}
\label{app:trapped}

The Penrose singularity theorem \cite{penrose1965gravitational} uses an initial condition, namely the presence of a trapped surface, or a surface whose mean normal curvature is negative
\be
\mathrm{H}^{\mu}U_{\mu} < 0 \,,
\ee
for all \(U^{\mu}\) future pointing null vectors.

It has been shown that a trapped surface lies inside the horizon of a classical black hole. In particular \cite{wald2010general}
\begin{prop}
\label{prop:trapclas}
Let $(M,g_{\mu \nu})$ be a strongly asymptotically predictable spacetime for which the null convergence condition holds. Suppose $M$ contains a trapped surface $T$. Then $T \subset B$, where $B$ is the black hole region of spacetime.
\end{prop}

The proof is straightforward as under the validity of the null convergence condition, the classical area theorem holds and any surface outside the horizon must fulfill \(\mathrm{H}^{\mu}U_{\mu} \ge 0\). With the possibility of an evaporating horizon this is not true anymore, as surfaces outside the horizon are allowed to shrink too. However, under the weaker energy conditions that allow an horizon to evaporate, the initial condition to prove a singularity theorem must be strengthened, namely that it is not enough to start from a trapped surface, but a ``sufficiently trapped'' surface is needed \cite{Freivogel:2020hiz}. Then a proposition similar to \ref{prop:trapclas} can be proven.

\begin{prop}
Let $(M,g_{\mu \nu})$ be a strongly asymptotically predictable spacetime. Let $T$ a co-dimension-$2$ spacelike hypersurface for which its mean normal curvature $H^\mu$ is sufficient to prove null geodesic incompleteness for null geodesics $\gamma$ emanating normally from $T$. Then $T \subset B$, where $B$ is the black hole region of spacetime.
\end{prop}

\begin{proof}
Focal points form along every normal null geodesic $\gamma$ emanating normally from $T$ if and only if
\be
U_\mu H^\mu \big|_{\gamma(0)} \leq -\frac{1}{n-2} \int_\gamma \big((n -2)f'(\lambda)^2 - f(\lambda)^2 R_{\mu \nu} U^\mu U^\nu \big)d\lambda \,,
\ee
holds for any \(U^{\mu}\). This is necessary to prove the contradiction to prove null geodesic incompleteness. However, outside the horizon the null generators of the boundary of the future of \(T\) must satisfy 
\be
U_\mu H^\mu  \geq -\frac{1}{n-2} \int_\gamma \big((n -2)f'(\lambda)^2 - f(\lambda)^2 R_{\mu \nu} U^\mu U^\nu \big)d\lambda \,.
\ee
The reason is that null generators are prompt curves, containing no focal points. Thus any surface \(T\) must lie inside the black hole region.
\end{proof}

This is a generalization of proposition \(12.2.2\) of~\cite{wald2010general}. However, in the classical case the null convergence condition is required. Here no energy condition is necessary: this is indeed only a statement of geometric properties of the spacetime.

Such a result is in accordance with the proof of singularity theorems for evaporating black holes \cite{Freivogel:2020hiz}. There, it was estimated how far inside the black hole is the sufficiently trapped surface to prove null geodesic incompleteness.

\section{Compatibility of pointwise condition with QEI}
\label{app:l0upperbound}

Independently from the area theorem, it is necessary to impose that the hypotheses required for theorem \ref{the:m=1} are consistent with each other. This comes naturally when \(Q_0 + \rho_0 \geq 0\), while it requires a bit more meditation when \(Q_0 + \rho_0 < 0\). 

In this second case, we consider $\lambda \in [0,\ell_0]$ where $R_{\mu \nu} U^\mu U^\nu \geq -|\rho_0|$. This pointwise condition allows $R_{\mu \nu} U^\mu U^\nu = -|\rho_0|$ for all $\lambda \in [0,\ell_0]$. But for large $\ell_0$ this can be incompatible with the QEI 
\be
\int_0^{\ell_0} g(\lambda)^2 R_{\mu\nu}U^{\mu}U^{\nu} \ge -Q_1(\gamma) \vert\vert g'\vert\vert^2 - Q_0(\gamma) \vert\vert g\vert\vert^2 \,.
\ee
that should hold in the same regime. So for $g$ compactly supported in $[0,\ell_0]$ we need to have
\be
-|\rho_0| ||g||^2 \ge -Q_0||g||^2 - Q_1 ||g'||^2 \,.
\ee
To check this in the case that $\rho_0+Q_0<0$ we study the functional \(K[g,g'] = -|Q_0 + \rho_0|||g||^2 + Q_1||g'||^2\). Imposing the boundary condition \(g(0) = 0\)) its Euler-Lagrange equation gives us the minimal functions 
\be
\tilde{g}(\lambda) = A\sin{\left(\lambda\sqrt{\frac{|Q_0 + \rho_0|}{Q_1}}\right)} \,,
\ee
where $A$ is an irrelevant constant. So the minimal values of \(K\) are:
\be
K\left[\tilde{g},\tilde{g}'\right] = Q_1\tilde{g}\tilde{g}'\Big\vert_{0}^{\ell_0} = \frac{A^2}{2}\sqrt{Q_1|Q_0 + \rho_0|}\sin{\left(2\ell_0\sqrt{\frac{|Q_0 + \rho_0|}{Q_1}}\right)} \,.
\ee
Therefore, to always have \(K\) non-negative we need
\be
\ell_0 \le \frac{\pi}{2}\sqrt{\frac{Q_1}{|Q_0 + \rho_0|}} \,,
\ee
which imposes an upper bound on $\ell_0$.

\newpage
\bibliographystyle{utphys}
\bibliography{bibliografia.bib}
%\printbibliography

\end{document}

%% file: newcommands.tex
\definecolor{darkturquoise}{rgb}{0.0, 0.81, 0.82}
\definecolor{cerisepink}{rgb}{0.93, 0.23, 0.51}
\definecolor{brilliantlavender}{rgb}{0.96, 0.73, 1.0}
\definecolor{fuchsiapink}{rgb}{1.0, 0.47, 1.0}

\newtheorem{theorem}{Theorem}[section]

\newtheorem{prop}[theorem]{Proposition}

\theoremstyle{definition}
\newtheorem{definition}{Definition}[section]

\theoremstyle{remark}
\newtheorem*{remark}{Remark}

\newcommand{\nord}[1]{{:}#1{:}}
\def\luv{\ell_{\rm UV}}
\DeclareMathOperator{\arcsinh}{arcsinh}

\newcommand{\be}{\begin{equation}}
\newcommand{\ee}{\end{equation}}
\newcommand{\bea}{\begin{eqnarray}}
\newcommand{\eea}{\end{eqnarray}}
\def\bml{\begin{subequations}}
\def\blea{\bml\begin{eqnarray}}
\def\eml{\end{subequations}}
\def\elea{\end{eqnarray}\eml}

\DeclareMathOperator{\csch}{csch}